\documentclass[letterpaper, 10 pt, conference]{ieeeconf}  

\IEEEoverridecommandlockouts                              

\usepackage{graphics} 
\usepackage{amsmath}
\usepackage{amssymb}
\usepackage{amsfonts}
\usepackage{mathtools}
\usepackage{makecell}
\usepackage{array}

\newtheorem{theorem}{Theorem}

\newtheorem{definition}{Definition}

\newtheorem{remark}{Remark}

\usepackage[
  style=ieee, 
  doi=false,       
  isbn=false,      
  url=false,       
  eprint=false     
]{biblatex}
\usepackage[short, nocomma]{optidef} 

\usepackage{balance} 
\usepackage{hyperref}
\usepackage{algorithm}
\usepackage{algorithmic}
\usepackage{booktabs}
\usepackage{multirow}
\usepackage{bm}
\usepackage{tikz}
\usetikzlibrary{arrows, arrows.meta, positioning, calc}
\usepackage{xcolor}

\title{\LARGE \bf
Learning Neural Network Controllers with Certified Robust Performance 
via Adversarial Training
}

\author{Neelay Junnarkar$^{1}$, Yasin Sonmez$^{1}$, Murat Arcak$^{1}$
\thanks{This work was supported in part by the NSF grant CNS-2111688.}
\thanks{$^{1}$Neelay Junnarkar, Yasin Sonmez, and Murat Arcak are with the Department of Electrical Engineering and Computer Sciences, University of California, Berkeley, CA 94720 USA (\texttt{neelay.junnarkar@berkeley.edu, yasin\_sonmez@berkeley.edu, arcak@berkeley.edu})}%
}

\begin{document}

\maketitle
\thispagestyle{empty}
\pagestyle{empty}

\begin{abstract}
Neural network (NN) controllers achieve strong empirical performance on nonlinear dynamical systems, yet deploying them in safety-critical settings requires robustness to disturbances and uncertainty.
We present a method for jointly synthesizing NN controllers and {dissipativity certificates} that formally guarantee robust closed-loop performance using adversarial training, in which we use counterexamples to the robust dissipativity condition to guide training.
Verification is done post-training using \(\alpha\),\(\beta\)-CROWN, a branch-and-bound-based method that enables direct analysis of the nonlinear dynamical system.
The proposed method uses quadratic constraints (QCs) only for characterization of non-parametric uncertainties.
The method is tested in numerical experiments on maximizing the volume of the set on which a system is certified to be robustly dissipative.
Our method certifies regions up to \(78\times\) larger than the region certified by a linear matrix inequality-based approach that we derive for comparison.

\end{abstract}
\section{Introduction}

Neural network (NN) controllers have shown remarkable performance in complex control tasks, but their deployment in safety-critical applications is hindered by the lack of formal guarantees on closed-loop behavior. 
Unlike classical controllers whose stability and robustness properties can be analyzed with well-established tools, NN controllers are nonlinear and high-dimensional, and hence difficult to certify.
While recent works have made progress in certifying nominal closed-loop stability for NNs \cite{yin_stability_2022, wangYoulaRENLearningNonlinear2022, pmlr-v235-yang24f, li2025two}, guaranteeing performance under external disturbances and model uncertainty remains a critical but less explored problem \cite{junnarkar2025stability}.

Dissipativity theory \cite{willems_dissipative_1972, arcak_networks_2016} provides a framework that generalizes Lyapunov stability to capture exactly these robustness requirements. By enforcing a dissipation inequality involving a storage function, which is analogous to a Lyapunov function, and a supply rate (e.g., $\ell_2$-gain, passivity), dissipativity enables the certification of input-output properties such as disturbance attenuation in addition to properties such as asymptotic stability.
We further consider robust dissipativity, which requires dissipativity to hold under model uncertainty.
Thus, robust dissipativity extends certificates from stability to robust performance under both external disturbances and model uncertainty.

A widely used method for analyzing such nonlinear and uncertain systems is the framework of integral quadratic constraints (IQCs) \cite{megretski_system_1997}, which abstracts nonlinearities with quadratic terms and renders conditions into computationally tractable linear matrix inequalities (LMIs) \cite{veenman_robust_2016, seiler_stability_2015}.
IQCs have been used to characterize known activation functions of NNs \cite{fazlyab_safety_2022, yin_stability_2022, pauli_training_2022, junnarkar2025stability}.
However, replacing known nonlinearities with bounding IQCs inherently introduces conservativeness. 
To overcome this, alternative techniques like Satisfiability Modulo Theories (SMT) \cite{abate_formal_synthesis, zhou_neural_2024} and branch-and-bound (BaB) methods such as $\alpha$,$\beta$-CROWN \cite{wang2021betacrown, pmlr-v235-yang24f, li2025two} have been utilized for tighter verification.
Adversarial training \cite{madry2018towards, goodfellow2015explaining} has been used in conjunction with these tools to guide synthesis.

The main contribution of this paper is a method to jointly synthesize controllers and certificates that ensure closed-loop robust dissipativity of systems subject to non-parametric uncertainty.
We leverage the complementary strengths of $\alpha$,$\beta$-CROWN and QCs: $\alpha$,$\beta$-CROWN is used for direct analysis of known nonlinear dynamics, and QCs are used for characterization of non-parametric model uncertainty.
We define a notion of robust dissipativity compatible with bounded sets, derive conditions verifiable via $\alpha$,$\beta$-CROWN, and propose an adversarial training algorithm for joint synthesis of the controller and dissipativity certificate.
Final verification is done using \(\alpha\),\(\beta\)-CROWN.
Benefits are demonstrated in two numerical examples in which, given a performance specification and a set describing the model uncertainty, controllers are trained to maximize the volume of the state space on which the closed-loop is certified to be robustly dissipative.

The rest of the paper is organized as follows.
Section~\ref{sec:problem-setup} gives background on modeling problems for verification with $\alpha$,$\beta$-CROWN, on quadratic constraints, and robust dissipativity.
Section~\ref{sec:main-results} derives conditions to be verified with $\alpha$,$\beta$-CROWN to ensure robust dissipativity of the closed-loop system, and derives analogous LMI conditions for comparison in numerical experiments.
Section~\ref{sec:methodology} details the proposed method for training controllers and verifying robust dissipativity.
Section~\ref{sec:experiments} demonstrates the method in two numerical experiments, including comparison with the LMI-based approach.

\subsection{Notation}
$\mathbb{R}_{\geq 0}$ denotes the set of nonnegative real numbers, $\ell_2^n$ the set of square-summable sequences with elements in $\mathbb{R}^n$, and  $\ell_{2e}^n$ the set of sequences with elements in $\mathbb{R}^n$ that are square-summable on 
$0, \dots, K$ for all $K \geq 0$.
We drop the superscript for the dimension of the codomain when clear from context.
\(\Omega_{V,\rho}\) denotes the sublevel set \(\{x|V(x) \leq \rho\}\) of \(V\).
\(B_r^n\) denotes the ball in \(\mathbb{R}^n\) of radius \(r\) centered at the origin.

\section{Background and Problem Setup} \label{sec:problem-setup}

\subsection{\(\alpha\),\(\beta\)-CROWN}

$\alpha$,$\beta$-CROWN enables verification of the specifications of the form \(f(x) > 0\) for all \(x \in \mathbb{R}^n\) in an interval \(x_\text{low} \leq x \leq x_\text{high}\) (inequalities taken elementwise) by computing global lower bounds \(\underline{f} \leq f(x)\). 
If \(\underline{f} > 0\), then the specification \(f(x) > 0\) is certified for all inputs in the domain.
The procedure is highly parallelizable and runs on GPUs via the autoLiRPA library~\cite{xu2020automatic}.
Various logical formulas can equivalently be written as \(f(x) > 0\) by defining operations as follows:
\begin{align*}
    \max\{x,y\} & = \frac{1}{2}\big(x + y \\
    & \hphantom{ = } + \mathrm{ReLU}(x-y) + \mathrm{ReLU}(y - x)\big) \\
    \min\{x,y\} & = -\max\{-x, -y\} \\
    x > 0 \land y > 0 & \iff \min\{x, y\} > 0 \\
    x > 0 \lor y > 0 & \iff \max\{x, y\} > 0 \\
    \lnot (x > 0) & \iff -x \geq 0 \\
    x > 0 \implies y > 0 & \iff \lnot (x > 0) \lor (y > 0) \\
\end{align*}
This enables specifications to be simply written as logical formulas.

\subsection{Dissipativity}

We consider performance requirements 
expressed
with a robust notion of dissipativity that considers model uncertainty.
We define robust dissipativity on subsets of state and external input spaces, allowing us to consider dissipativity on bounded sets that $\alpha$,$\beta$-CROWN can verify.

Consider the system
\begin{equation}\label{eq:diss-sys}
\begin{aligned}
    x_{k+1} & = F(x_k, w_k, d_k) \\
    v_k & = G(x_k, w_k, d_k) \\
    e_k & = H(x_k, w_k, d_k) \\
    w_k & = \Delta(v)_k
\end{aligned}
\end{equation}
where \(x_k \in \mathbb{R}^n\) is the system state, \(v_k \in \mathbb{R}^{n_v}\) and \(w_k \in \mathbb{R}^{n_w}\) are the inputs/outputs of the uncertainty \(\Delta\), \(d_k \in \mathbb{R}^{n_d}\) is an exogenous input, \(e_k \in \mathbb{R}^{n_e}\) is a performance output, \(k \geq 0\), and the uncertainty \(\Delta\) belongs to a set of operators \(\mathbf{\Delta}\).
Assume the system is well-posed: for any initial condition \(x_0 \in \mathbb{R}^n\), \(d \in \ell_{2e}\), and \(\Delta \in \mathbf{\Delta}\), there exist unique sequences \(x, v, w,\) and \(e\) satisfying the system equations.

\begin{definition}[Robust Dissipativity]\label{def:dissipativity}
    Given \(\mathcal{X} \subseteq \mathbb{R}^n\), \(\mathcal{D} \subseteq \mathbb{R}^{n_d}\), and the class of uncertainty $\mathbf{\Delta}$,
    the system \eqref{eq:diss-sys} is robustly dissipative on \((\mathcal{X}, \mathcal{D}, \mathbf{\Delta})\) with respect to a supply rate \(s(d, e)\) if the following hold:
    \begin{enumerate}
        \item If \(\Delta \in \mathbf{\Delta}, x_0 \in \mathcal{X}\), and \(d_k \in \mathcal{D}\) for \(k \geq 0\), then \(x_k \in \mathcal{X}\) for \(k \geq 0\).
        \item There exists a storage function \(V: \mathcal{X} \to \mathbb{R}_{\geq 0}\) with \(V(0) = 0\) such that if \(\Delta \in \mathbf{\Delta}, x_0 \in \mathcal{X}\), and \(d_k \in \mathcal{D}\) for \(k \geq 0\), then \(V(x_K) - V(x_0) \leq \sum_{k=0}^{K-1} s(d_k, e_k)\) for \(K \geq 1\).
    \end{enumerate}
\end{definition}
The first condition is a robust forward invariance (RFI) condition, and the second condition is the performance condition expressed as a dissipation inequality~\cite{arcak_networks_2016}.
Some typical supply rates are \(\gamma^2 \|d\|^2 - \|e\|^2\), which corresponds to an \(\ell_2\)-gain bound of \(\gamma\), and \(d^\top e\), which corresponds to passivity.

\subsection{Quadratic Constraints}

An operator \(\Delta: \ell_{2e} \to \ell_{2e}\) satisfies the quadratic constraint (QC) defined by \(M = M^\top\) if 
\begin{equation*}
    \begin{bmatrix}
        v_k \\ w_k
    \end{bmatrix}^\top 
    M
    \begin{bmatrix}
        v_k \\ w_k
    \end{bmatrix}
    \geq 0
\end{equation*}
for all \(v \in \ell_{2e}\), \(k \geq 0\), and \(w = \Delta(v)\).
QCs of this form can describe properties including sector bounds, gain margins, and multiplication by matrices in a polytope \cite{megretski_system_1997, lessard_analysis_2016}.

Given a set \(\mathcal{M}\) of symmetric matrices, we define the uncertainty set \(\mathbf{\Delta}(\mathcal{M})\) as the set of operators that satisfy the QC defined by \(M\) for all \(M \in \mathcal{M}\).
The simplest family of QCs is \(\{\lambda M_0 | \lambda \geq 0\}\) for a fixed \(M_0 = M_0^\top\).

\section{Main Results}\label{sec:main-results}

We present sufficient conditions, that can be verified with $\alpha$,$\beta$-CROWN and with LMIs, to ensure robust dissipativity on bounded sets \(\mathcal{X}\) and \(\mathcal{D}\) for a set of uncertainties \(\mathbf{\Delta}(\mathcal{M})\).

\begin{theorem}\label{thm:dissipativity}
    Let \(\mathcal{M}\) be a set of QCs, \(V: \mathbb{R}^{n} \to \mathbb{R}_{\geq 0}\) be such that \(V(0) = 0\), and \(\rho > 0\).
    Then \eqref{eq:diss-sys} is robustly dissipative with respect to a supply rate \(s(d, e)\) on \((\Omega_{V,\rho}, B_{\overline{d}}^{n_d}, \mathbf{\Delta}(\mathcal{M}))\) if there exist \(M_\text{rfi}, M_\text{perf} \in \mathcal{M}\) such that
    \begin{align}
        z^\top M_\text{rfi} z \geq 0 \implies V(F(x, w, d)) & \leq \rho \label{eq:pointwise-rfi}\\
        z^\top M_\text{perf} z \geq 0 \implies V(F(x, w, d)) - V(x) & \leq s(d, e) \label{eq:pointwise-sr}
    \end{align}  
    for all \(x \in \Omega_{V,\rho}, w \in \mathbb{R}^{n_w}, d \in B_{\overline{d}}^{n_d}\) where \(z = (G(x, w, d), w)\) and \(e = H(x, w, d)\).
\end{theorem}
\begin{proof}
    Consider a trajectory of the system with \(\Delta \in \mathbf{\Delta}(\mathcal{M})\) such that \(V(x_0) \leq \rho\) and \(d_k \in B_{\overline{d}}^{n_d}\) for \(k \geq 0\).  
    Because \(w = \Delta(v)\), it holds that \(z_k^\top M z_k \geq 0\) for all \(k \geq 0\) for
    all \(M \in \mathcal{M}\).
    This implies, by \eqref{eq:pointwise-rfi} and induction, that \(x_k \in \Omega_{V,\rho}\) for all \(k \geq 0\).
    Therefore, \eqref{eq:pointwise-sr} applies for each time step \(k \geq 0\) of the trajectory.
    Summing the right-hand side of the implication in \eqref{eq:pointwise-sr} over \(k = 0, \dots, K-1\) implies \(V(x_K) - V(x_0) \leq \sum_{k=0}^{K-1} s(d_k, e_k)\) for \(K \geq 1\).
    Therefore both conditions of Definition~\ref{def:dissipativity} are met, and \eqref{eq:diss-sys} is robustly dissipative on \((\Omega_{V,\rho}, B_{\overline{d}}, \mathbf{\Delta}(\mathcal{M}))\).
\end{proof}
\begin{remark} \label{rmk:local-qc}
    If the QCs in \(\mathcal{M}\) hold locally in the region
    \(\mathcal{V} = \bigcap_i \left\{(x,w,d) | v^\top L_i^\top L_i v \leq \overline{v}_i^2\right\}\) for some matrices (e.g. selection matrices) \(L_i\) and bounds \(\overline{v}_i \in \mathbb{R}\) where \(v = G(x, w, d)\), then we must ensure \(\Omega_{V,\rho} \subseteq \mathcal{V}\).
    This containment holds if there exist \(M_{\Delta i} \in \mathcal{M}\) such that
    \begin{equation}\label{eq:local-containment}
        z^\top M_{\Delta i} z \geq 0 \implies v^\top L_i^\top L_i v \leq \overline{v}_i^2
    \end{equation}
    for all \(x \in \Omega_{V,\rho}, w \in \mathbb{R}^{n_w}\), \(d \in B_{\overline{d}}^{n_d}\), and \(L_i\), where \(z = (G(x, w, d), w)\).
    

\end{remark}

\subsection{Verification using $\alpha$,$\beta$-CROWN} \label{sec:verification-with-crown}

When \(\Omega_{V,\rho}\) is bounded, the robust forward invariance and performance conditions in Theorem~\ref{thm:dissipativity} are verified using $\alpha$,$\beta$-CROWN by writing them as conditions \(\varphi_\text{rfi}(x, w, d) > 0\) and \(\varphi_\text{perf}(x, w, d) > 0\).
The condition \(\varphi_\text{rfi}\) is constructed by first writing \eqref{eq:pointwise-rfi} as a logical formula, where \(\xi = ({x}, w, d)\):
\begin{equation*}
\begin{aligned}
    & \left((V(x) \leq \rho)  \land (d^\top d \leq  \overline{d}^2) \land (z^\top M_\text{rfi} z \geq 0) \land (\xi^\top \xi \geq \epsilon)\right) \\
    & \implies -V(F({x}, w, d)) + \rho \geq 0
\end{aligned}
\end{equation*}
Note that $\alpha$,$\beta$-CROWN verifies strict inequalities, so we exclude a small region around the origin using \(\xi^\top \xi \geq \epsilon\) and \(\epsilon \approx 10^{-3}\).
This formula is converted into a function \(\mathbb{R}^{n} \times \mathbb{R}^{n_w} \times \mathbb{R}^{n_d} \to \mathbb{R}\) using the constructions in Section~\ref{sec:problem-setup}.
A similar procedure is taken to construct \(\varphi_\text{perf}\).

The boxes for \(x\) and \(d\) over which these conditions are verified are selected to contain \(\Omega_{V,\rho}\) and \(B_{\overline{d}}\), respectively.
We describe the method of determining the box for \(x\) in Section~\ref{sec:methodology}.
The box for \(d\) is \(\{d | \|d\|_\infty \leq \overline{d}\}\).

To construct \(w\)---which is in general unbounded---for use in $\alpha$,$\beta$-CROWN, we leverage the properties of the QCs that \(\Delta\) satisfies to construct \(w\) through a differentiable transformation of \(x\), \(d\), and auxiliary parameters.
Whether such a transformation exists depends on the QC and the structure of \(G\).
If \(\Delta\) satisfies a gain bound \(\gamma^2 \|v_k\|^2 - \|w_k\|^2 \geq 0\) and \(v_k\) can be expressed in terms of \(x_k\) and \(d_k\), then we parameterize \(w_k\) as \(w_k = (\tilde{w}_1 \gamma \|v_k\|/\|\tilde{w}_2\|) \tilde{w}_2\) where \(\tilde{w}_1 \in [-1, 1]\) and \(\|\tilde{w}_2\|_\infty \leq 1\).
This can be generalized to \(G\) that are affine in \(w\) under some additional conditions.
We label the domain of the parameters \(\tilde{w}\) of \(w\) as \(\tilde{\mathcal{W}}\).
Now the conditions \(\varphi_\text{rfi}\) and \(\varphi_\text{perf}\) have domain \(\mathcal{X} \times \tilde{\mathcal{W}} \times \mathcal{D}\).
We refer to the conditions \(\varphi_\text{rfi}\) and \(\varphi_\text{perf}\) 
as functions of \(x\), \(w\), and \(d\) for simplicity of notation.

Parameterizations of \(w\) such as the one presented above ensure \(z^\top M z \geq 0\) by construction, so the term \(z^\top M z \geq 0\) may be removed from \(\varphi_\text{rfi}(x, w, d)\) and \(\varphi_\text{perf}(x, w, d)\).
For more general sets of QCs \(\mathcal{M}\), it may be necessary to assume a bound on \(w\) such that \(w^\top w \leq \overline{w}^2\), sample from the box containing this ball, and integrate this condition into Theorem~\ref{thm:dissipativity} in a manner similar to the bound on disturbance.

\subsection{Verification using LMIs} \label{sec:verification-with-lmis}

For comparison with the use of \(\alpha\),\(\beta\)-CROWN, we now present a QC-based method to analyze robust stability, characterizing both uncertainties and nonlinearities with QCs.
Consider the reformulation of \eqref{eq:diss-sys} as the interconnection of an LTI system and nonlinearities and uncertainties gathered in \(\tilde{\Delta}\):
\begin{equation}\label{eq:uncertain-lti-sys}
\begin{aligned}
    {x}_{k+1} & = A {x}_k + B_w \tilde{w}_k + B_d d_k \\
    \tilde{v}_k & = C_v {x}_k + D_{v w} \tilde{w}_k + D_{v d} d_k \\
    e_k & = C_e {x}_k + D_{e w} \tilde{w}_k + D_{e d} d_k \\
    \tilde{w}_k & = \tilde{\Delta}(\tilde{v})_k
\end{aligned}
\end{equation}
The set of operators \(\tilde{\Delta}\) considered will now be defined in terms of a set of QCs \(\tilde{\mathcal{M}}\), which is constructed from \(\mathcal{M}\) to characterize the original \(\Delta\) and additional QCs to characterize the nonlinearities.

We now provide conditions amenable to analysis via semidefinite programming that are sufficient for Theorem~\ref{thm:dissipativity}.
\begin{theorem} \label{thm:lmi-dissipativity}
    Let \(\tilde{\mathcal{M}}\) be a set of QCs,
    \(P \in \mathbb{R}^{n \times n}\) be such that \(P \succ 0\), and \(\rho > 0\).
    Consider a quadratic supply rate \(s(d, e)\).
    If there exist \(s_\rho \geq 0, s_d \geq 0\), and \(M_\text{rfi}, M_\text{perf} \in \tilde{\mathcal{M}}\) such that \eqref{eq:lmi-pointwise-rfi} and \eqref{eq:lmi-pointwise-sr} hold, then \eqref{eq:diss-sys} is robustly dissipative with respect to \(s\) on \((\Omega_{V,\rho}, B_{\overline{d}}^{n_d}, \mathbf{\Delta}({\mathcal{M}}))\).
    
    In the following, let \(X, Z, D, S\) be such that \(\xi = ({x}, \tilde{w}, d)\), \(z = (\tilde{v}, \tilde{w})\), \({x}_{k+1} = X\xi_k\), \(z_k = Z\xi_k\), \(\begin{bmatrix}
        d_k^\top & e_k^\top
    \end{bmatrix}^\top = D\xi_k\) and \(s(d_k, e_k) = \xi_k^\top D^\top S D \xi_k\).
    \begin{subequations} \label{eq:lmi-pointwise-rfi}
    \begin{align}
        -X^\top P X + \begin{bmatrix}
        s_\rho P & 0 & 0 \\ 0 & 0 & 0 \\ 0 & 0 & s_d I
        \end{bmatrix} - Z^\top M_\text{rfi} Z & \succeq 0 \label{eq:lmi-pointwise-rfi-matrix}\\
        (1 - s_\rho) \rho  - s_d \overline{d}^2 & \geq 0 \label{eq:lmi-pointwise-rfi-scalar}
    \end{align}
    \end{subequations}
    \begin{equation}\label{eq:lmi-pointwise-sr}
        \begin{aligned}
            -X^\top P X + \begin{bmatrix}
            P & 0\\ 0 & 0
        \end{bmatrix}
         - Z^\top M_\text{perf} Z + D^\top S D& \succeq 0
        \end{aligned}
    \end{equation}
\end{theorem}
\begin{proof}
    Define \(V(x) = x^\top P x\).
    Label the matrix in \eqref{eq:lmi-pointwise-rfi-matrix} as \(Y\) and the scalar in \eqref{eq:lmi-pointwise-rfi-scalar} as \(s\).
    By condition \eqref{eq:lmi-pointwise-rfi}, \(Y \succeq 0\) and \(s \geq 0\), which is equivalent to \(\xi^\top Y \xi + s \geq 0\) for all \(\xi\).
    Expanding, this is equivalent to \(-V(F(x, w, d)) + \rho - s_\rho (-V(x) + \rho) - s_d(-d^\top d + \overline{d}^2) - z^\top M_\text{rfi} z \geq 0\) for all \(x, w, d\).
    This is a sufficient condition for \eqref{eq:pointwise-rfi}.
    Similarly, left- and right-multiplying the matrix in \eqref{eq:lmi-pointwise-sr} by \(\xi^\top\) and \(\xi\) shows it is equivalent to \(-V(F(x,w,d)) + V(x) + s(d, e) - z^\top M_\text{perf} z \geq 0\) for all \(x, w, d\).
    This is a sufficient condition for \eqref{eq:pointwise-sr}.
    Therefore, by Theorem~\ref{thm:dissipativity}, \eqref{eq:diss-sys} is robustly dissipative on \((\Omega_{V,\rho}, B_{\overline{d}}, \mathbf{\Delta}({\mathcal{M}}))\).    
\end{proof}
\begin{remark} \label{rmk:lmi-local-qc}
    If the QCs in \(\tilde{\mathcal{M}}\) hold locally, then the following provides a sufficient condition for \eqref{eq:local-containment}, where \(\tilde{V}_i\) is such that  \(\tilde{L}_i \tilde{v} = \tilde{V}_i \xi\).
    \begin{align*}
        -\tilde{V}_i^\top \tilde{V}_i + \begin{bmatrix}
            s_{\rho i}P & 0 & 0 \\ 0 & 0 & 0 \\ 0 & 0 & s_{di}I
        \end{bmatrix}-Z^\top M_{\Delta i} Z & \succeq 0 \\
        \overline{\tilde{v}}_i^2 - s_{\rho i}\rho - s_{di}\overline{d}^2 & \geq 0
    \end{align*}
    
\end{remark}

Given a particular \(P\), \(\overline{d}\), and \(\tilde{\mathcal{M}}\), the problem of finding the largest sub-level set of \(V(x) = x^\top P x\) such that \eqref{eq:diss-sys} is dissipative on \((\Omega_{V,\rho}, B_{\overline{d}}, \mathbf{\Delta}(\tilde{\mathcal{M}}))\) is formulated as the following optimization problem:
\begin{maxi!}[2]
    {\rho, s_\rho, s_d, M_\text{rfi}, M_\text{perf}}{\rho}
    {}{}
    \addConstraint{s_\rho, s_d}{\geq 0}
    \addConstraint{M_\text{rfi}, M_\text{perf}}{\in \tilde{\mathcal{M}}}
    \addConstraint{\eqref{eq:lmi-pointwise-rfi}, \eqref{eq:lmi-pointwise-sr}}{}
\end{maxi!}
This is an SDP for fixed \(\rho\).
To address bilinearity in \(\rho\) and decision variable \(s_\rho\), we bisect on \(\rho\).
If quadratic constraints \(\tilde{\mathcal{M}}\) are local, variables and constraints are used as per Remark~\ref{rmk:lmi-local-qc}.
It is often useful to parameterize the local QC in terms of the region over which it is valid, and search over this set of local QCs.
We do this by sampling different local QCs and, for each, determining the maximum \(\rho\) by bisection, and taking the maximum of all \(\rho\) found.

\section{Methodology} \label{sec:methodology}

In this section, we detail our algorithm for training controllers and verifying robust dissipativity of systems of the form \eqref{eq:diss-sys} using $\alpha$,$\beta$-CROWN.
We verify the condition \((\varphi_\text{rfi}(x, w, d) > 0) \land (\varphi_\text{perf}(x, w, d) > 0)\), which we label \(\varphi(x, w, d)\) (see Section~\ref{sec:verification-with-crown} for construction of the robust forward invariance performance requirement conditions).
Algorithm~\ref{alg:main} summarizes the training and verification pipeline. 
It has two stages: the first trains learnable parameters---the controller, storage function, and others---so that \(\varphi\) is satisfied and the certified volume is maximized, and the second uses $\alpha$,$\beta$-CROWN to verify the result of the training.

\begin{algorithm}[htbp]
\caption{Adversarial Training and Verification}
\label{alg:main}
\begin{algorithmic}[1]
\REQUIRE System \eqref{eq:diss-sys}, supply rate $s$, quadratic constraints $\mathcal{M}$, initial controller \(\pi_0\), storage function initialization $P \succ 0$, initial bounding box for state $\mathcal{B}_0$, disturbance box $\mathcal{D}$
\STATE \textbf{Training:}
\STATE Initialize \(P_\phi \gets P, V_\phi^{(0)} \gets x^\top P_\phi x, \pi_\theta^{(0)} \gets \pi_0\)\;
\STATE Run PGD + $\alpha$,$\beta$-CROWN bisection on $V_\phi^{(0)} \to \rho_0$
\STATE Sample anchors $\mathcal{A}$ in $V_\phi^{(0)}(x) \in [\alpha_{\mathrm{exp,1}}\rho_0,\; \alpha_{\mathrm{exp,2}}\rho_0]$
\FOR{epoch $= 1, \ldots, N_{\mathrm{epochs}}$}
    \STATE $\tilde{\xi}_i = (x_i, \tilde{w}_i, d_i) \sim \mathrm{Uniform}(\mathcal{B}_j\times \mathcal{\tilde{W}} \times \mathcal{D})$
    \STATE $\tilde{\xi}_i^\text{adv} \gets \text{PGD to estimate}\ \arg\min_{\tilde{\xi}} \varphi(\tilde{\xi})$
    \STATE $\rho \leftarrow \min_{x \in \partial \mathcal{B}_j} V_\phi(x)$
    \STATE $ \mathcal{L}_\text{total} \gets \lambda_\text{s}\mathcal{L}_\text{sample} + \lambda_\text{adv}\mathcal{L}_\text{adv} + \lambda_\text{anc}\mathcal{L}_\text{anchor}$
    \STATE Update $\theta,\phi$, other learnable parameters using \(\mathcal{L}_\text{total}\)
    \IF{$\min_i \varphi(x_i, w_i, d_i) \geq 0$ for $N_{\mathrm{clean}}$ consecutive epochs}
        \STATE Simulate trajectories from probe box $\supset \mathcal{B}_j$; keep converging ones
        \STATE $\mathcal{B}_{j+1} \leftarrow$ bounding box of reachable envelopes (capped by $\eta_{\mathrm{max}}$)
        \STATE \textbf{break} (end training) if $\mathcal{B}_{j+1} = \mathcal{B}_j$
    \ENDIF
\ENDFOR
\STATE \textbf{Verification:} Bisect over $\rho$ with $\alpha$,$\beta$-CROWN (Sec.~\ref{sec:verification})
\RETURN Certified region $\Omega_{V_\phi, \rho^*}$, controller \(\pi_\theta\)
\end{algorithmic}
\end{algorithm}

\subsection{Controller and Certificate}

\subsubsection{Storage Function}

We use a neural network storage function that combines a quadratic term and a nonlinear modulation.
Let $P_\phi \coloneqq \epsilon I + R^\top R$ where $R \in \mathbb{R}^{n \times n}$ is a trainable matrix and $\epsilon > 0$ is a fixed constant, and let $\psi_\phi \colon \mathbb{R}^n \to \mathbb{R}$ be a  neural network.
The storage function $V_\phi\colon \mathbb{R}^n \to \mathbb{R}$ is defined as
\begin{equation}\label{eq:storage}
    V_\phi(x) = \underbrace{x^\top P_\phi\, x}_{V_{\mathrm{quad}}(x)} \;\cdot\; \underbrace{\big(1 + \alpha_\text{nn}\,\tanh\!\big(\psi_\phi(x) - \psi_\phi(0)\big)\big)}_{m_\phi(x)},
\end{equation}
where $x$ is measured relative to the equilibrium $x^*$ and $\alpha_{nn} \in (0,1)$ is a fixed scaling hyperparameter.
The function $V_\phi$ is nonnegative and satisfies \(V_\phi(0) = 0\) by construction.
Further, \(V_\phi\) is positive definite and radially unbounded because \(m_\phi(x) \geq 1 - \alpha_{nn}\) for all \(x\), ensuring \(\Omega_{V_\phi,\rho}\) is compact for all \(\rho\).

We initialize this storage function with a quadratic storage function \(x^\top P x\) obtained by the LMI method.
We find that normalizing \(P\) by its Frobenius norm improves numerical stability during training.
At initialization, the parameters of the last layer of \(\psi_\phi\) are set to 0 so that \(V_\phi(x) = x^\top P_\phi x\).
The Cholesky decomposition yields the \(R\) matrices to initialize \(P_\phi\).
This initialization provides a warm start that accelerates convergence.
During training, both \(R\) and the network \(\psi_\phi\) are updated.

\subsubsection{Controller} \label{sec:controller}

The proposed method is compatible with standard neural network and controller architectures such as multi-layer perceptrons, linear state-feedback, and LTI controllers.
In this paper, we use a recurrent implicit neural network (RINN)~\cite{junnarkar2025stability} due to its compatibility with LMI methods for robustness analysis, which enables us to use LMI methods to construct initializing controllers.
Note that Algorithm~\ref{alg:main} can also be used to verify systems with a fixed controller by fixing its parameters.

A RINN controller is a generalization of an LTI controller constructed by interconnecting an LTI system with an implicit neural network~\cite{el_ghaoui_implicit_2021}.
It is of the form
\begin{equation*}
\pi_\theta \left\{
\begin{aligned}
    \dot{x}_K(t)  & = A_K x_K(t) + B_{Kw} w_K(t) + B_{Ky} y(t) \\
    v_K(t) & = C_{Kv} x_K(t) + D_{Kvw}w_K(t) + D_{Kvy}y(t) \\
    u(t) & = C_{Ku} x_K(t) + D_{Kuw} w_K(t) + D_{Kuy} y(t) \\
    w_K(t) &= \sigma(v_K(t))
\end{aligned}
\right.
\end{equation*}
where \(x_K \in \mathbb{R}^{n_K}\) is the controller state, \(w_K \in \mathbb{R}^{n_{Kw}}\) is the output of the implicit neural network, \(v_K \in \mathbb{R}^{n_{Kw}}\), \(y\) is the output of the plant, \(u\) is the control input of the plant, and \(\sigma\) is the activation function, which we take to be \(\mathrm{ReLU}\).
To simplify the evaluation of the controller, we restrict \(D_{Kvw}\) to be strictly upper triangular.


\subsubsection{QC Multipliers and Supply-Rate Scale}

The conditions of Theorem~\ref{thm:dissipativity} allow for any QC defined by \(M\) in a set \(\mathcal{M}\).
Given a set \(\mathcal{M}\) parameterized through a differentiable transformation of parameters \(\Lambda\), we add these parameters to the set of parameters that are tuned through gradient descent in training.
For example, for a set \(\mathcal{M} = \{\lambda M_0 | \lambda \geq 0\}\), we introduce and train the parameter \(\lambda \geq 0\) (parameterized as \(\lambda = \tilde{\lambda}^2\) where \(\tilde{\lambda} \in \mathbb{R}\)) with \(\lambda\) initialized to 1.0.

Further, the system \eqref{eq:diss-sys} is dissipative with respect to a supply rate \(s\) if and only if it is dissipative with respect to \(\alpha_s s\) where \(\alpha_s > 0\).
To automatically tune this supply rate scale hyperparameter, we introduce it as a learnable parameter (parameterized as \(\alpha_s = \exp(\beta_s)\) with \(\beta_s \in \mathbb{R}\)).
In the dissipation inequality condition \(\varphi_\text{perf}\), we use the supply rate \(\alpha_s s\).
To initialize \(\alpha_s\), we use \(1/\|P\|_F\), where \(P\) is the initial storage function certificate found for certifying the dissipation inequality and \(P/\|P\|_F\) is the normalized matrix used to initialize \(P_\phi\).

\subsection{Training Procedure} \label{sec:training}

We train the controller parameters, the storage function parameters, and the multipliers to satisfy \(\varphi(x, w, d) > 0\) through adversarial training.

\subsubsection{Sub-level Set and Loss}

Given a box \(\mathcal{B}\) for the state, we estimate \(\rho = \min_{x \in \partial \mathcal{B}} V_\phi(x)\), the largest sub-level set of \(V_\phi\) that is contained in \(\mathcal{B}\), through projected gradient descent (PGD)~\cite{goodfellow2015explaining} on each face of \(\mathcal{B}\). PGD is an iterative first-order method that repeats the cycle of taking a gradient step and projecting back onto a constraint set; it was introduced for evaluating adversarial robustness of neural networks~\cite{madry2018towards} and provides an efficient inner maximizer (or minimizer) over bounded domains.

We then attempt to train the learnable parameters to verify the condition \(\varphi\) on \(\Omega_{V_\phi, \rho}\). In each training epoch, a batch of state--disturbance pairs $(x_i, \tilde{w}_i,  d_i)$ is drawn uniformly at random from the bounding box $\mathcal{B} \times \mathcal{\tilde{W}} \times \mathcal{D}$ and \(w\) is computed from \((x_i, \tilde{w}_i, d_i)\) (see Section~\ref{sec:verification-with-crown}).
The training loss penalizes violations of the verification condition at these samples:
\begin{equation}\label{eq:loss-sample}
    \mathcal{L}_\text{sample} = \frac{1}{N_\text{sample}}\sum_{i=1}^{N_\text{sample}} \mathrm{ReLU}\!\big(\!-\varphi(x_i, w_i, d_i)\big).
\end{equation}

\subsubsection{Adversarial Augmentation Loss}

Uniform sampling alone may miss thin violation regions where $\varphi < 0$, since such regions can occupy a negligible fraction of the domain volume.
To address this, we augment the training set with PGD adversarial examples: starting from random initializations in $\mathcal{B} \times \tilde{\mathcal{W}} \times \mathcal{D}$, PGD minimizes $\varphi(x, w, d)$ subject to $(x, \tilde{w}, d) \in \mathcal{B} \times \mathcal{\tilde{W}} \times \mathcal{D}$, concentrating samples where violations are most likely.
Worst-case points are stored in a ranked replay buffer.
We compute a second loss term using adversarial points sampled from this buffer.
\begin{equation}\label{eq:loss-adv}
    \mathcal{L}_\text{adv} = \frac{1}{N_\text{adv}}\sum_{i=1}^{N_\text{adv}} \mathrm{ReLU}\!\big(\!-\varphi(x_i^\text{adv}, w_i^\text{adv}, d_i^\text{adv})\big).
\end{equation}

\subsubsection{Growth Incentive via Anchor Sampling Loss}

To encourage the sub-level set \(\Omega_{V_\phi, \rho}\) to grow,
we sample \emph{anchor} states in a band around the initial verified region and penalize high storage-function values at those states.
Specifically, let $\rho_0$ denote the initial PGD-verified sub-level set value and $V_\phi^{(0)}$ the initial storage function.
We sample anchor states whose initial storage-function values lie in $[\alpha_{\mathrm{exp},1}\,\rho,\; \alpha_{\mathrm{exp},2}\,\rho]$, where $\alpha_{\mathrm{exp},1} < 1$ penalizes for shrinking below the initial certified region and $\alpha_{\mathrm{exp},2} > 1$ rewards expansion beyond it.
The regularization term is
\begin{equation}\label{eq:loss-anchor}
    \mathcal{L}_{\mathrm{anchor}} = \frac{1}{|\mathcal{A}|}\sum_{x \in \mathcal{A}} \mathrm{ReLU}\!\left(\frac{V_\phi(x)}{\rho} - 1\right),
\end{equation}
where $\mathcal{A}$ is the set of anchor states and $\rho$ is the current sub-level set value.
This term penalizes anchor states that lie outside the current sublevel set, encouraging the optimizer to reshape $V_\phi$ so that $\Omega_{V_\phi, \rho}$ expands to include them.
Hyperparameters are listed in Table~\ref{tab:params}.

The total loss is then a weighted sum of individual loss terms.
\begin{equation}\label{eq:total-loss}
    \mathcal{L}_\text{total} = \lambda_\text{s}\mathcal{L}_\text{sample} + \lambda_\text{adv}\mathcal{L}_\text{adv} + \lambda_\text{anc}\mathcal{L}_\text{anchor}
\end{equation}

\begin{figure*}[htbp]
\centering
\includegraphics[width=0.7\textwidth]{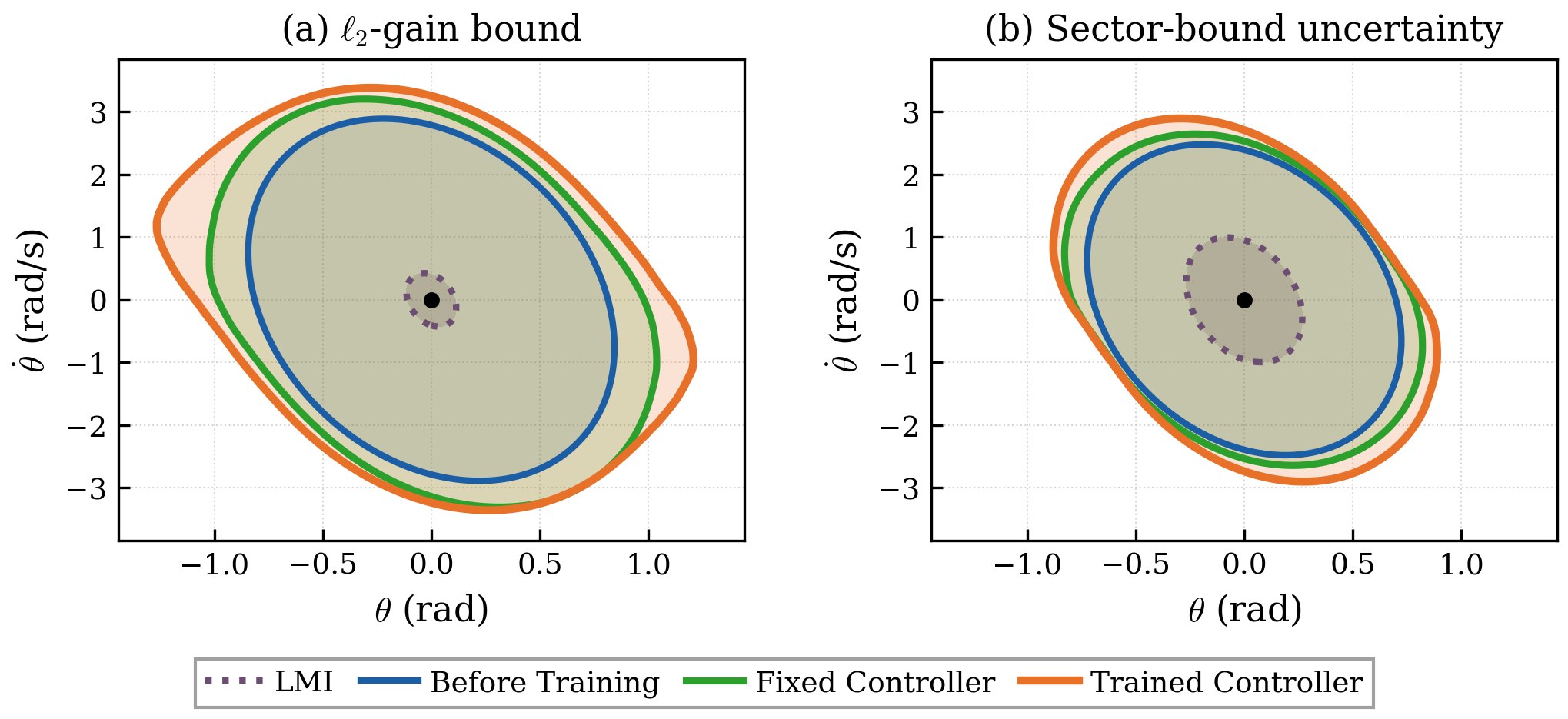}
\caption{Certified regions projected onto the plant state $(\theta, \dot\theta)$ for (a) the $\ell_2$-gain bound experiment and (b) the robust stability experiment under sector-bound uncertainty ($\alpha = 0.25$). Four methods are compared: LMI baseline (purple, dotted), before training (blue), training only the storage function with a fixed controller (green), and jointly training both controller and storage function (orange). 
}
\label{fig:roa}
\end{figure*}

\subsubsection{Simulation-Based Domain Expansion}

We increase the size of \(\mathcal{B}\), and therefore the size of the sub-level set we aim to verify, using trajectory simulations, following the approach of~\cite{li2025two}.
Expansion is triggered when PGD finds no violations for $N_\mathrm{clean}$ consecutive epochs, indicating that the adversarial sampler can no longer falsify the condition over the current region.
At each expansion step, we sample initial conditions from a \emph{probe box} that extends beyond the current domain $\mathcal{B}_j$, simulate the system for $K_\mathrm{sim}$ steps, and retain only those trajectories that converge to the origin (measured by the terminal state norm falling below a threshold $\delta_\mathrm{conv}$).
The new domain $\mathcal{B}_{j+1}$ is the element-wise bounding box of the reachable envelopes of the converging trajectories, capped by a maximum growth factor $\eta_\mathrm{max}$.
This procedure expands the domain where the current controller stabilizes the system, while avoiding regions where it does not.
Training terminates when the domain no longer grows (no converging trajectories reach beyond the current box).

\subsection{Formal Verification} \label{sec:verification}

After training, we fix the trained parameters and use $\alpha$,$\beta$-CROWN to verify $\varphi_\text{rfi}(x,w,d) > 0$ and $\varphi_\text{perf}(x, w, d) > 0$ through bisection on the level set \(\rho\) to maximize the volume of the verified region. 
First, given the final box \(\mathcal{B}\) on state from training, we estimate $\rho_{\max} = \min_{x \in \partial \mathcal{B}} V_\phi(x)$.
Then, the bisection procedure is as follows.
\begin{enumerate}
    \item Initialize the search bracket: first test $\rho_0 = \beta \rho_\text{max}$ where $\beta < 1$ (we use $\beta \sim 0.95$).
    If $\alpha$,$\beta$-CROWN verifies $\rho_0$: bracket $= [\rho_0, \mu\rho_\text{max}]$.
    If not: bracket $= [\rho_\text{max}/\mu, \rho_0]$ where $\mu > 1$ is a multiplier.
    \item At each bisection step, PGD pre-screens the candidate $\rho$; if a counterexample is found, the candidate is rejected without calling $\alpha$,$\beta$-CROWN.
    \item $\alpha$,$\beta$-CROWN verifies the sub-problem $\varphi_\text{rfi}(x, w, d) > 0$. If it fails, the upper bound is lowered, and we go to step (2).
    \item $\alpha$,$\beta$-CROWN verifies the sub-problem  $\varphi_\text{perf}(x, w, d) > 0$. If it passes, the lower bound is raised; otherwise, the upper bound is lowered.
    \item Repeat from step (2) until the bracket width is below a tolerance $\delta$.
\end{enumerate}

\begin{remark}[Tight bounding boxes]
For each candidate $\rho$, we compute a tight bounding box for $\Omega_{V_\phi, \rho}$ by sampling and finding the element-wise extrema of points satisfying $V(x) \leq \rho$, expanded by a small margin.
This replaces the full domain $\mathcal{B}$ in the specification and tightens the $\alpha$,$\beta$-CROWN bounds.
\end{remark}

As a baseline, we run the full bisection procedure on the initial quadratic storage function $V_\phi^{(0)}(x) = x^\top P_\phi\, x$ (with zero neural correction) with the initial controller.
Comparing these pre-training certified regions to the post-training results isolates the improvement due to the nonlinear storage function and controller training.

\section{Numerical Experiments} \label{sec:experiments}

We evaluate the proposed framework on a torque-actuated inverted pendulum\footnote{\url{https://github.com/neelayjunnarkar/local-robust-dissipativity}}.
In two experiments, we present two results using the method presented in Section~\ref{sec:methodology}: one where we train both the controller and dissipativity certificates to maximize the certified volume, and a second where we train only the dissipativity certificates and leave the controller fixed.
Where applicable, we compare the results with the regions certified using the LMI techniques presented in Section~\ref{sec:verification-with-lmis}.

\subsection{Certifying \(\ell_2\)-Gain Bound}  \label{sec:exp-l2gain}

Consider the following model of an inverted pendulum with a disturbance entering into the control input.
\begin{equation*}
\begin{aligned}
\dot{x}_{P1}(t) &= x_{P2}(t) \\
\dot{x}_{P2}(t)  &= -\frac{\mu}{m \ell^2}x_{P2}(t) + \frac{g}{\ell}\sin(x_{P1}(t)) \\
& \hphantom{=.} + \frac{1}{m\ell^2}(\mathrm{sat}_{\overline{u}}(u(t))+d(t)) \\
y(t) & = x_P(t)
\end{aligned}
\end{equation*}
The control input \(u\) is saturated in the interval \([-\overline{u}, \overline{u}]\) and the disturbance is bounded in absolute value by \(0.1 \overline{u}\).
All parameters are given in Table~\ref{tab:params}.

We certify dissipativity of the closed-loop of the plant and controller from disturbance \(d\) to performance output \(e = x_P\) with respect to the \(\ell_2\) gain supply rate \(s(d, e) = \gamma^2 \|d\|^2 - \|e\|^2\).
We initialize the RINN controller with one synthesized using LMI methods similar to \cite{junnarkar2025stability}.
We use no saturation on the controller, no bounds on the disturbance, and a quadratic constraint on \(\sin\) that is valid for \(x_{P1} \in [-\pi, \pi]\).
This comes with an associated quadratic storage function.
The plant and controller are discretized individually using Euler integration, and Theorem~\ref{thm:dissipativity} is applied to the closed-loop system, which has state \(x = (x_P, x_K)\).

We compute four certified regions \(\Omega_{V, \rho}\): 

\begin{enumerate}
    \item Using the initial controller and its associated storage function using the LMI method in Section~\ref{sec:verification-with-lmis}.
    \item Using the initial controller and its associated storage function using $\alpha$,$\beta$-CROWN as in Section~\ref{sec:verification}.
    \item  Using the initial controller and training the certificates with Algorithm~\ref{alg:main}.
    \item Training both the controller and the certificates with Algorithm~\ref{alg:main}.
\end{enumerate}

For the LMI method, we apply local sector constraints to \(\sin\) and to the saturation function \(\mathrm{sat}(v) = \min\{1, \max\{-1, v\}\}\) (the function \(\mathrm{sat}_{\overline{u}}(u) = \overline{u} \ \mathrm{sat}(u/\overline{u})\)).
We grid search over local sector bounds on \(\sin\) valid for \(|x_{P1}| \leq \overline{x}_{P1}\) with \(\overline{x}_{P1} \in [0, \pi]\) and similarly grid search over local sector bounds on \(\mathrm{sat}\) valid for \(|v| \leq \overline{v}\) with \(\overline{v} \in [1.0, 5.0]\) using a resolution of 0.1.

Table~\ref{tab:combined-results} reports the volumes of the certified sub-level sets \(\Omega_{V,\rho}\) projected onto the plant state \(x_P\), and the computation times.
Figure~\ref{fig:roa} visualizes the projections of $\Omega_{V_\phi, \rho}$ onto the two-dimensional plane of plant states, given as:
\begin{equation} \label{eq:proj-roa}
    \mathcal{P}(\Omega_{V, \rho}) = \{x_P | \min_{x_K} V(x_P, x_K) \leq \rho\}
\end{equation}
and is estimated by minimizing \(V\) over the controller states at each point of a fine grid.
All experiments are run on a computer with a Nvidia RTX 5090 and an AMD 9950 X3D.


\begin{table}[!htbp]
\centering
\caption{Volumes and Computation Times}
\label{tab:combined-results}
\setlength{\tabcolsep}{1pt}
\begin{tabular}{@{}l 
>{\centering\arraybackslash}p{1.1cm} 
>{\centering\arraybackslash}p{1.1cm} 
>{\centering\arraybackslash}p{1.2cm} 
c c c@{}}
\toprule
& \multicolumn{3}{c}{\textbf{Volumes} $\mathcal{P}(\Omega_{V,\rho})$ ($\mathrm{rad{\cdot}rad/s}$)} & \multicolumn{3}{c}{\textbf{Times} (min.)} \\
\cmidrule(lr){2-4} \cmidrule(l){5-7}
\textbf{Method} & \textbf{Abs.}
  & \textbf{vs.\ LMI}
  & \textbf{vs.\ Before train}
  & \textbf{Train} & \textbf{Verif.} & \textbf{Total} \\
\midrule
\multicolumn{7}{c}{\textbf{$\ell_2$-Gain Bound}} \\
\midrule
1) LMI                &  $0.16    $  & $1.0\times$   & ---           & N/A   & $0.03$ & $0.03$ \\
2) Before Training    & $8.2$  & $51.2\times$  & $1.0\times$   & N/A   & 9      & 9      \\
3) Fixed Controller   & $10.8$ & $67.5\times$  & $1.32\times$  & 3     & 33     & 36     \\
4) Trained Controller & $12.5$ & $78.1\times$  & $1.52\times$  & 6     & 26     & 32     \\
\midrule
\multicolumn{7}{c}{\textbf{Robust Stability}} \\
\midrule
1) LMI                & $0.80   $ & $1.0\times$   & ---           & N/A   & $0.03$ & $0.03$ \\
2) Before Training    & $6.0$  & $7.6\times$   & $1.0\times$   & N/A   & $21$   & $21$   \\
3) Fixed Controller   & $7.1$  & $8.9\times$   & $1.17\times$  & $7$   & $99$   & $106$  \\
4) Trained Controller & $8.1$  & $10.2\times$  & $1.35\times$  & $19$  & $96$   & $115$  \\
\bottomrule
\end{tabular}%
\end{table}

The LMI baseline yields the smallest certified volume.
Even before training, verifying the same closed-loop system using the same storage function with $\alpha$,$\beta$-CROWN gives a $51{\times}$ larger region, illustrating the conservatism of using QCs to characterize the nonlinearities.
Training the storage function increases the verified volume to $67.5{\times}$ the LMI baseline and $1.32{\times}$ the region verified with $\alpha$,$\beta$-CROWN without training.
Joint training of the controller and the storage function results in the largest verified volume at $78.1{\times}$ the LMI baseline, $1.52{\times}$ the region verified with $\alpha$,$\beta$-CROWN without training, and $1.16{\times}$ the region obtained by training the storage function alone.
This demonstrates that substantial enlargement of the certified region is achievable even without modifying the controller, and that co-training provides additional benefit.


A note on the gap between the volume that is not falsified by PGD attacks and the volume that is certified with $\alpha$,$\beta$-CROWN:
With the initial controller and storage function, PGD and $\alpha$,$\beta$-CROWN areas are approximately identical ($8.2$ vs.\ $8.2$).
After training, a modest gap appears (e.g., $13.1$ vs.\ $12.5$ for the trainable case), reflecting the inevitable relaxation in $\alpha$,$\beta$-CROWN's bound propagation through the more expressive trained storage function.


\subsection{Robust Stability under Model Uncertainty} \label{sec:exp-robust}

\begin{figure}[htbp]
\centering
\begin{tikzpicture}[auto, >=latex', 
    block/.style={draw, rectangle, minimum height=1.0cm, minimum width=1.2cm, inner sep=5pt},
    sum/.style={draw, circle, minimum size=0.4cm, inner sep=0pt},
    dot/.style={circle, draw, fill=black, inner sep=0pt, minimum size=3pt}
]
    \tikzset{every node/.style={font=\small}}
    \node [sum] (sum) at (0,0) {$+$};
    \node [block] (plant) at (2.2, 0) {$P$};
    \node [block] (delta) at (0, 1.5) {$\Delta$};
    \node [name=u] at (-2.8, 0) {$\tilde{u}$};
    \coordinate (branch_pt) at (-1.5, 0); 
    \node [left] at (-1.5, 0.75) {$\tilde{u}$}; 
    \node [right] at (3.5, 0) {$x$}; 
    \draw [->] (sum.east) -- node[above] {$\tilde{u}+w$} (plant.west);
    \draw [->] (plant.east) -- (3.5, 0);
    \draw [->] (delta.south) -- node[right] {$w$} (sum.north);
    \draw [-] (u.east) -- (branch_pt);
    \node [dot] at (branch_pt) {};
    \draw (branch_pt) -- (-1.5, 1.5);
    \draw [->] (-1.5, 1.5) -- (delta.west);
    \draw [->] (branch_pt) -- (sum.west);
\end{tikzpicture}
\caption{Uncertainty structure: Plant input is sum of control \(\tilde{u}\) and \(w = \Delta(\tilde{u})\), where \(\Delta\) is uncertain and satisfies \(\|w\| \leq \alpha \|\tilde{u}\|\) pointwise in time.}
\label{fig:uncertainty_structure}
\end{figure}
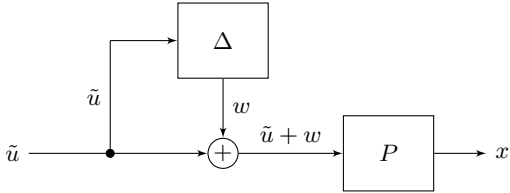

We now apply the framework to verify robust stability (as in, use supply rate \(s(d,e) = 0\)) for the same pendulum
subject to non-parametric uncertainty $\Delta$ on the input channel (Figure~\ref{fig:uncertainty_structure}).
The uncertain dynamics take the form
\begin{equation*}
\begin{aligned}
\dot{x}_{P1}(t) &= x_{P2}(t) \\
\dot{x}_{P2}(t)  &= -\frac{\mu}{m \ell^2}x_{P2}(t) + \frac{g}{\ell}\sin(x_{P1}(t)) \\
& \hphantom{=.} + \frac{1}{m\ell^2}\big(\tilde{u}(t)+ \Delta(\tilde{u}(t))\big) \\
\tilde{u}(t) & = \mathrm{sat}_{\overline{u}}(u(t)) \\
y(t) & = x_P(t)
\end{aligned}
\end{equation*}
where $\Delta$ is characterized by the sector-bound QC
\begin{equation}\label{eq:sector-qc}
    \begin{bmatrix} v \\ w \end{bmatrix}^\top
    \underbrace{\begin{bmatrix} \alpha^2 & 0 \\ 0 & -1 \end{bmatrix}}_{M_0}
    \begin{bmatrix} v \\ w \end{bmatrix} \geq 0,
\end{equation}
so that $\|w(t)\| \leq \alpha \|u(t)\|$ for all \(t\).
We use the same initializing controller and storage function.

Note that verifying stability of this uncertain system implies that, among other uncertainties, the controller stabilizes the plants \(g P\) for all \(g \in [1 - \alpha, 1 + \alpha]\) where \(P\) is the nominal plant with \(\Delta = 0\).
Therefore, this uncertainty is sufficient for analyzing gain margin conditions.

To construct a bounded parameterization suitable for $\alpha$,$\beta$-CROWN, we write $w_k = \alpha\, \tilde{w}_k\, v_k, \quad \tilde{w}_k \in [-1, 1]$, which satisfies~\eqref{eq:sector-qc} by construction.
The verification conditions $\varphi_\text{rfi}$ and $\varphi_\text{perf}$ are then defined over the bounded domain $\mathcal{X} \times [-1,1]$.
Results are reported in Table~\ref{tab:combined-results}.
The LMI baseline certifies $0.80\,\mathrm{rad{\cdot}rad/s}$ with the uncertainty characterized by a sector-bound QC.
Before training, $\alpha$,$\beta$-CROWN certifies a projected area of $6.0\,\mathrm{rad{\cdot}rad/s}$ ($7.6{\times}$ the LMI baseline), smaller than the $\ell_2$-gain case due to the cost of certifying stability for all plants in the uncertainty set.
Training the storage function with a fixed controller increases the $\alpha$,$\beta$-CROWN-verified area to $7.1$ ($+17\%$ over the initial certificate).
Joint training of the controller and storage function yields $8.1\,\mathrm{rad{\cdot}rad/s}$ ($+35\%$ over the initial, $1.14{\times}$ the fixed-controller result).
Figure~\ref{fig:roa} visualizes the certified regions.

\begin{table}[htbp]
\centering
\caption{Physical and algorithm parameters.}
\label{tab:params}
\resizebox{\columnwidth}{!}{
\begin{tabular}{@{}llcc@{}}
\toprule
\textbf{Parameter} & \textbf{Symbol} & \boldmath$\ell_2$\textbf{-Gain Bound} & \textbf{Robust stability} \\
\midrule
\multicolumn{4}{@{}l}{\emph{Plant (shared)}} \\
Mass & $m$ & \multicolumn{2}{c}{$0.15$\,kg} \\
Length & $\ell$ & \multicolumn{2}{c}{$0.5$\,m} \\
Damping & $\mu$ & \multicolumn{2}{c}{$0.1$\,Nm\,s/rad} \\
Gravity & $g$ & \multicolumn{2}{c}{$9.81$\,m/s$^2$} \\
Time step & $\Delta t$ & \multicolumn{2}{c}{$0.01$\,s} \\
Torque limit & $\overline{u}$ & \multicolumn{2}{c}{$0.75$\,Nm ($\approx 1.02\,mgl$)} \\
\midrule
\multicolumn{4}{@{}l}{\emph{Supply rate / uncertainty}} \\
$\ell_2$-gain bound & $\gamma$ & $100$ & --- \\
Disturbance bound & $\overline{d}$ & $0.075$\,Nm  & --- \\
Uncertainty bound & $\alpha$ & --- & $0.25$ \\
Auxiliary variable & $\tilde{w}$ & --- & $\in [-1,\,1]$ \\
Supply rate & $s(d,e)$ & $\gamma^2\|d\|^2 {-} \|e\|^2$ & $0$ \\
\midrule
\multicolumn{4}{@{}l}{\emph{Controller (shared)}} \\
RINN internal states & $n_K$ & \multicolumn{2}{c}{$2$} \\
RINN implicit nodes & $n_{Kw}$ & \multicolumn{2}{c}{$8$} \\
\midrule
\multicolumn{4}{@{}l}{\emph{Storage function}} \\
Storage NN hidden & --- & $[128,\,128]$ & $[32,\,32]$ \\
Hidden activation & --- & \multicolumn{2}{c}{LeakyReLU} \\
Neural scale & $\alpha_{\mathrm{nn}}$ & $0.5$ & $0.25$ \\
Domain (initial) & $\mathcal{B}_0$ & \multicolumn{2}{c}{$[-3,3] {\times} [-9,9] {\times} [-4,4]^2$} \\
\midrule
\multicolumn{4}{@{}l}{\emph{Anchor sampling}} \\
Anchor count & $|\mathcal{A}|$ & \multicolumn{2}{c}{$1{,}024$} \\
Inner anchor factor & $\alpha_{\mathrm{exp},1}$ & \multicolumn{2}{c}{$0.75$} \\
Outer anchor factor  & $\alpha_{\mathrm{exp},2}$ & $1.3$ / $1.2$ (train/fix) & $1.3$ / $1.1$ (train/fix) \\
Anchor weight & $\lambda_a$ & \multicolumn{2}{c}{$0.1$} \\
Loss weights & \(\lambda_\text{s}, \lambda_\text{adv}\) & \multicolumn{2}{c}{$1$}\\
\bottomrule
\end{tabular}
}
\end{table}

\section{Conclusion} \label{sec:conclusion}

This paper presented a framework for training neural network controllers and corresponding certificates to ensure robust dissipativity of the closed-loop system.
A key advantage of this approach is its ability to leverage the complementary strengths of $\alpha$,$\beta$-CROWN for direct analysis of known nonlinear dynamics and QCs for non-parametric model uncertainty.
As demonstrated in numerical experiments, this method significantly reduces conservativeness compared to LMI-based techniques, yielding substantially larger certified volumes, although at the expense of higher computation time.
Future work includes extending the framework from QCs to integral quadratic constraints.

\balance

\renewcommand*{\bibfont}{\normalfont\small}
\printbibliography

\end{document}